\documentclass[aps,pra,twocolumn,superscriptaddress,floatfix,nofootinbib,showpacs,longbibliography]{revtex4-1}
\pdfoutput=1
\usepackage[utf8]{inputenc}  
\usepackage[T1]{fontenc}     
\usepackage[british]{babel}  
\usepackage[colorlinks=true, citecolor=blue, urlcolor=blue]{hyperref}  
\usepackage{graphicx} 
\usepackage[babel]{microtype}  
\usepackage{amsmath,amssymb,amsthm,bm,amsfonts,mathrsfs,bbm} 

\usepackage{xspace}  
\usepackage{pgf,tikz}
\usepackage{xcolor}
\usepackage{multirow}
\usepackage{array}
\usepackage{bigstrut}
\usepackage{braket}
\usepackage{color}
\usepackage{natbib}
\usepackage{multirow}
\usepackage{mathtools}
\usepackage[normalem]{ulem}
\usepackage{float}
\usepackage[caption = false]{subfig}
\usepackage{xcolor,colortbl}
\usepackage{color}
\usepackage{tikz}
\usetikzlibrary{positioning,arrows.meta,calc}
\usepackage{amssymb}
\newcommand{\Tr}{\operatorname{Tr}}

\newcommand{\be}{\begin{equation}}
\newcommand{\ee}{\end{equation}}
\newcommand{\ba}{\begin{eqnarray}}
\newcommand{\ea}{\end{eqnarray}}
\newcommand{\ketbra}[2]{|#1\rangle \langle #2|}

\newtheorem{theorem}{Theorem}
\newtheorem{corollary}{Corollary}

\newtheorem{proposition}{Proposition}

\newtheorem{lemma}{Lemma}
\newenvironment{theorem'}
 {\expandafter\def\expandafter\thetheorem\expandafter{\thetheorem'}\theorem}
 {\endtheorem}
\usepackage{mathtools}

\begin{document}

\title{
Free encoding capacity: A universal unit for quantum resources}
\author{Shampa Mondal}
\affiliation{Physics and Applied Mathematics Unit, Indian Statistical Institute, Kolkata, India}

\author{Soumajit Das}
\affiliation{Physics and Applied Mathematics Unit, Indian Statistical Institute, Kolkata, India}

\author{Preeti Parashar}
\affiliation{Physics and Applied Mathematics Unit, Indian Statistical Institute, Kolkata, India}

\author{Tamal Guha}
\affiliation{Mathematical Institute, Slovak Academy of Sciences, Bratislava, Slovakia}

\begin{abstract}
A perfect \(d\)-dimensional quantum channel can convey \(\log d\)-bits of classical information by encoding messages in \(d\) orthogonal quantum states. Alternatively, for every quantum state at the senders end, there exist \(d\) encoding operations which produce \(d\)-orthogonal quantum states. Transmitting which via a \(d\)-level perfect quantum channel it is possible to communicate \(\log d\)-bits of classical information. But what if the set of encoding operations is restricted only within a physically constrained class? Here, we consider such a class of encoding operations to be the set of free operations for any quantum resource theory and show that the constrained capacity --- namely, the \textit{free encoding capacity} (FEC) emerged as a unit of the corresponding quantum resource. Moreover, we show that for the \textit{pointed} resource theories --- a resource theory admitting only a single free state --- FEC becomes a faithful resource measure also. We also discuss the implications of FEC in the question of resource-theoretic state transformations and the possibility of extending its faithfulness for general quantum resource theories. 
\end{abstract}
\maketitle
\textit{Introduction.}-- The strength of quantum systems in classical information processing tasks relies on the non-classical elements, often referred to as the \textit{resources}, possessed by the theory itself. For instance, by exploiting entanglement, quantum systems remarkably enhance the information processing ability of its classical counterpart \cite{bennett1992communication}. Such an enhancement indeed  persists even when the quantum transmission lines are highly affected by noise \cite{bennett1999entanglement, chiribella2025communication}. Despite these advantages, quantum transmission lines face several no-go results in classical information processing \cite{shor2002additivity, king2001capacity, bennett2002entanglement, mozes2005deterministic, shirokov2012conditions}. Interestingly, by limiting the resource consumption to encode classical information in quantum states -- specifically with limited entanglement \cite{zhu2017superadditivity}, or limited energy \cite{mondal2025classical} -- one can overcome some of these no-go results further. Constraining the encoding resources was traditionally encountered within the continuous variable quantum systems to bypass the diverging physical quantities involved with the infinite dimensions \cite{braunstein2000dense, ban2000quantum, holevo2001evaluating, sen2005capacities}. Nevertheless, these suggest an inherent connection between the generic quantum resources and their information-theoretic implications. However, a concrete formalism for arbitrary quantum resources remains unexplored. An arbitrary quantum resource theory (QRT), in a nutshell, is characterized in terms of free states, a set of states possessing no resource content, and a class of operations -- the free operations -- employing which one cannot increase the amount of resource possessed by any quantum state \cite{chitambar2018quantum}. In principle, every operational signature of quantum theory -- ranging from nonlocality to the computational advancements -- bearing no or a less efficient classical counterpart can be regarded as a quantum resource and accordingly possesses a resource-theoretic framework \cite{horodecki2009quantum, brandao2013resource, veitch2014resource, streltsov2017structure, wu2021operational, ghosh2024quantum}.

Within such a notion of QRT, here we consider an elegant classical information processing scenario where the sender (say, Alice) is only allowed to implement the free operations of the corresponding QRTs to encode her classical message on a given quantum state and send it to the receiver (say, Bob). This leads to a restricted classical capacity of the given quantum state under the shared quantum channel, namely the \textit{free encoding capacity} (FEC). In our analysis, we have considered the transmission line -- carrying the classical information encoded in quantum states -- to be a perfect one, ensuring no resource consumption during the communication process itself. Precisely, while a perfect quantum channel can be considered as a maximally resourceful element in the context of classical information processing, from the perspective of QRT it is indeed a free channel (operation). This is because under the perfect quantum channel, i.e., an identity operation, the resource content of a given quantum system remains invariant. 
As a dual of the present communication scenario, one may consider a situation where Alice cannot directly access the quantum systems to encode her classical information rather is allowed to configure the transmission line she shares with Bob. For instance, in an optical-fiber based communication system, Alice -- without having access to the single photon generator directly -- can rearrange the retarder plates, polarizers, analyzers etc. to tune the polarization of the transmitting photon in accordance with the input classical information. Such a dual-encoding picture, from the perspective of quantum thermodynamics, emerged as the Helmholtz free energy of input quantum memory \cite{narasimhachar2019quantifying}. Nevertheless, here we show that the FEC emerges as a bounded, strongly monotonic and convex resource measure for any arbitrary QRT. In fact, to achieve the FEC in the asymptotic limit one may restrict within the set of extreme free operations only. This renders a universal resource unit for any QRT, arising from a seemingly uncorrelated direction of classical information processing. Additionally, the strong monotonic nature of such a resource measure then provides an upper bound on the probability of freely transforming a given quantum resource to the other. Moreover, for any arbitrary quantum resource theories with a single free state, the FEC becomes a faithful measure also. Such a class of QRT, namely the \textit{pointed} ones, often appeared in the realm of quantum information theory \cite{PhysRevA.67.062104, GOUR20151}, quantum thermodynamics \cite{brandao2015second} and in the identification of unknown quantum operations \cite{parisio2024quantum}. Finally, we give a prescription to redefine FEC to be a faithful resource measure for some specific classes, along with characterizing the structure of QRTs where it fails.

\textit{Framework.}-- Any arbitrary QRT \(\mathcal{R}\), motivated by an operational task, characterizes a class of states \(\rho\in\mathcal{D}(\mathcal{H}_d)\) for which the operational utility is zero or no better than its classical counterpart (\(\mathcal{D}(\mathcal{H}_d)\)denotes the set of density matrices over the Hilbert space \(\mathcal{H}_d\)). This set of states, denoted as \(\mathcal{F}_{\mathcal{R}}\), is convex in general, with some exceptions in the context of continuous variable systems \cite{turner2025all}, quantum correlations, and quantum dynamics \cite{salazar2024quantum}. By denoting an arbitrary utility function --- often referred to as the resource measure --- as \(\mathcal{U}_{\mathcal{R}}\), we can identify 
\[\mathcal{F}_{\mathcal{R}}:=\{\sigma\in\mathcal{D}(\mathcal{H}_d)|\mathcal{U}_{\mathcal{R}}(\sigma)=0\}.\]
The proper characterization of the set \(\mathcal{F}_{\mathcal{R}}\) then naturally asks about the stability of the set \(\mathcal{F}_{\mathcal{R}}\) under quantum operations. Accordingly, a set of quantum channels, \textit{completely positive trace preserving} (CPTP) maps over the linear operators \(\mathcal{L}(\mathcal{H}_d)\), can be characterized for which the set \(\mathcal{F}_{\mathcal{R}}\) remains invariant -- namely the free operations \(\mathbf{N}_{f}^{\mathcal{R}}\). Mathematically, it is defined as
\[\mathbf{N}_f^{\mathcal{R}}:=\{\mathcal{N}|\mathcal{N}(\sigma)\in\mathcal{F}_{\mathcal{R}},\forall \sigma\in\mathcal{F}_{\mathcal{R}}\}.\]
As an illustrative example, one could consider the resource theory of athermality \cite{brandao2013resource} -- motivated by the operational premise of extracting thermodynamic work from a quantum system \(\rho\in\mathcal{D}(\mathcal{H}_d)\), in terms of the Helmholtz free energy \(F_{\beta}(\rho)\) with respect to the given thermal bath of inverse temperature \(\beta\). The free state, in this case, is a thermal state (\(\tau_{\beta}\)) of the bath and the set of Gibbs-preserving operations \cite{faist2015gibbs}, which leave \(\tau_{\beta}\) invariant, can be characterized as free operations. One could, in principle, define an even stronger class of free operations whose Stinespring realizations involves only the free state \(\tau_{\beta}\) of the environment as an ancillary system. For the aforementioned resource theory of athermality the thermal operations are precisely such a stronger class of free operations. 

The resource theory of athermality belongs to a special class of QRTs, for which there is a single free state and here we termed them as \textit{pointed} QRTs. Besides their mathematical elegance, pointed QRTs possess several salient features such as existence of resource-destroying maps \cite{liu2017resource}, the precise one-shot distillable and dilution costs \cite{PhysRevLett.123.020401}. However, for such QRTs a standard notion of resource measure, namely the relative entropy with respect to the free state often fails to attain finiteness. For instance, recently a significant interest has been devoted to identifying non-uniformity of quantum state for a given basis -- namely quantum \textit{texture} \cite{parisio2024quantum} -- ranging its implications from field-atom interactions \cite{huang2025quantum} to quantum phase transition \cite{patra2025rolequantumstatetexture}. While there are various proposed measures to quantify the amount of texture a quantum state possesses \cite{zhang2025quantum, wang2025quantifying}, a central problem is that since the unique free state has rank one, QRT of texture does not admit traditional relative entropy as a finite measure of resource. At this juncture, we will now introduce a new way to measure such resources, which can be faithfully integrated for any arbitrary forthcoming pointed resource-theoretic signature of quantum theory.

Switching over the premises, we will now revisit the notion of elegant classical information processing via quantum transmission lines, a.k.a., quantum channels. Viewing any localized quantum operation in a time-like setting gives rise to a quantum channel. In general, quantum channels are CPTP maps acting between the linear operators over two different Hilbert spaces \(\mathcal{L}(\mathcal{H}_{d_1})\to\mathcal{L}(\mathcal{H}_{d_2})\); however, in our analysis it is sufficient to restrict within those acting between the operators of the isomorphic Hilbert spaces, i.e.,  \(\mathcal{H}_{d_1}\simeq\mathcal{H}_{d_2}\). This is because, from the resource-theoretic perspective, a free state of a given dimension may not be free in general for any arbitrary dimension. For instance, consider the maximally mixed qubit state; while the state is free for two-dimensional resource theory, it is no longer free for higher dimensions. 

Consider a classical information processing task, where the sender Alice possesses a \(d\)-level quantum state \(\rho\in\mathcal{D}(\mathcal{H}_d)\), along with an arbitrarily constrained set of CPTP maps \(\mathbf{N}_{R}\) as the allowed operations to encode classical information \(X\equiv\{x_i\}_{i=0}^{n-1}\) on \(\rho\). On the other hand, Bob can perform an arbitrary POVM \(\mathcal{M}\equiv\{M_y|M_y\geq0,~\sum_y M_y=\mathbb{I}_d\}\) on each of the individual transmitted quantum states and accordingly decode the classical random variable \(\{y_j\}_{j=0}^{n-1}\equiv Y\). Then the restricted classical capacity of the \(d\)-level quantum state \(\rho\) can be written as
\begin{align}\label{e1}
    \mathcal{C}^{(1)}_R(\rho)=\max_{p_x, \mathcal{N}_x\in\mathbf{N}_{R}, \mathcal{M}}I(X:Y),
\end{align}
where, \(I(X:Y)=H(X)+H(Y)-H(XY)\) is the Mutual Information between the random variables \(X\) and \(Y\), with \(H(Z)=-\sum_zp_z\log_2 p_z\) is the standard Shannon entropy of the random variable \(Z\) following a probability distribution \(\{p_z\}_z\) (from now on, we will use "\(\log\)" to mean logarithm with base \(2\) and hence all the information quantities will be measured in terms of "bits"). Within such a restricted set of encoding operations, one could define the restricted classical capacity of the perfect \(d\)-level quantum channel (\(\mathcal{I}_d\)) as 
\begin{align*}
    \mathcal{C}^{(1)}_{R}(\mathcal{I}_d)=\max_{\rho\in\mathcal{L}(\mathcal{H}_d)}\mathcal{C}^{(1)}_{R}(\rho).
\end{align*}
Nevertheless, here we will only focus on the restricted classical capacity of the given quantum state, instead of the channel.

Notably, for both capacities, the superscript \((1)\) denotes the simplest possible classical communication set-up involving product encoding on Alice's side and product decoding (i.e., the measurement on each particle individually) on Bob's side. However, if Bob performs a joint measurement on all the quantum states he obtained, then the regularized classical capacity of the quantum state \(\rho\) becomes the Holevo information of the ensemble \(\{p_x,\mathcal{N}_x(\rho)\}\) maximized over \(\{p_x, \mathcal{N}_x\in\mathbf{N}_R\}\), that is,
\begin{align}\label{e2}
\nonumber\mathcal{C}_R(\rho)&=\lim_{n\to\infty}\frac1n\mathcal{C}^{(n)}_R(\rho^{\otimes n})\\\nonumber&=\max_{p_x,\mathcal{N}_x\in\mathbf{N}_R}\chi(\sum_xp_x\ketbra{x}{x}\otimes\mathcal{N}_x(\rho))
\\\nonumber&=\max_{p_x,\mathcal{N}_x\in\mathbf{N}_R} [S(\sum_xp_x\mathcal{N}_x(\rho))-\sum_x p_xS(\mathcal{N}_x(\rho))]
\\&=\max_{p_x,\mathcal{N}_x\in\mathbf{N}_R}\mathcal{C}_{R}(\rho,\{p_x,\mathcal{N}_x\}),
\end{align}
where, \(\mathcal{S}(\sigma)=-\Tr(\sigma \log \sigma)\) is the von Neumann entropy of the quantum state \(\sigma\) and \(\chi\) denotes the Holevo information of a classical-quantum state. In the rest of the paper, we will use the notation \(\mathcal{C}_R(\rho, \{p_x,\mathcal{N}_x\})=\chi(\sum_xp_x\ketbra{x}{x}\otimes\mathcal{N}_x(\rho))\) to denote the communication utility of a quantum state \(\rho\), under the free encoding ensemble \(\{p_x,\mathcal{N}_x\}\).

In particular, when  the set of encodings \(\mathbf{N}_R\) is limited within the set of free operations (\(\mathbf{N}_{\mathcal{R}}\)) of an arbitrary QRT \(\mathcal{R}\) (we change the subscript \(R\) to \(\mathcal{R}\) corresponding to the QRT \(\mathcal{R}\)), then we will call the quantity \(\mathcal{C}_{\mathcal{R}}(\rho)\) as the \textit{Free Encoding Capacity} (FEC) of the given quantum state \(\rho\). 

\textit{Main Results.}-- With all the preliminaries, we have discussed so far, we can now simply state our main result in terms of the following theorem.
\begin{theorem}\label{t1}
    For any pointed resource theory \(\mathcal{R}\) and for any arbitrary quantum state \(\rho\), the free encoding capacity \(\mathcal{C}_{\mathcal{R}}(\rho)\) is a bounded, strongly monotonic, convex and faithful measure of the resource.
\end{theorem}
The proof of Theorem \ref{t1} can be seen as the concatenations of the following propositions. In particular, we will now show that the quantity \(\mathcal{C}_{\mathcal{R}}(\rho)\) respects all the individual features (boundedness, monotonicity, strong monotonicity under stochastically free operations, convexity, and faithfulness) for any arbitrary pointed QRT \(\mathcal{R}\). Moreover, the first three features indeed hold for any arbitrary QRT \(\mathcal{R}\). Before going to the proof of the theorem, we derive an elegant and significant corollary from it.
\begin{corollary}\label{c1}
    If \(\Pr(\rho\xrightarrow[\mathcal{R}]{}\sigma)\) denotes the probability of transforming an arbitrary quantum state \(\rho\) to another state \(\sigma\), under the stochastic free operations allowed in the QRT \(\mathcal{R}\), then
    \[\Pr(\rho\xrightarrow[\mathcal{R}]{}\sigma)\leq \min\Bigl\{ \frac{\mathcal{C}_{\mathcal{R}}(\rho)}{\mathcal{C}_{\mathcal{R}}(\sigma)}, 1\Bigr\}\]
\end{corollary}
\noindent
We defer the proof to Appendix \ref{pc1}. Note that such an upper-bound for stochastic state transformation probabilities has been individually encountered in various resource-theoretic contexts earlier \cite{baumgratz2014quantifying, wu2020quantum, wu2021resource}.

\begin{figure}[t]
    \centering
    \includegraphics[width=1.0\linewidth]{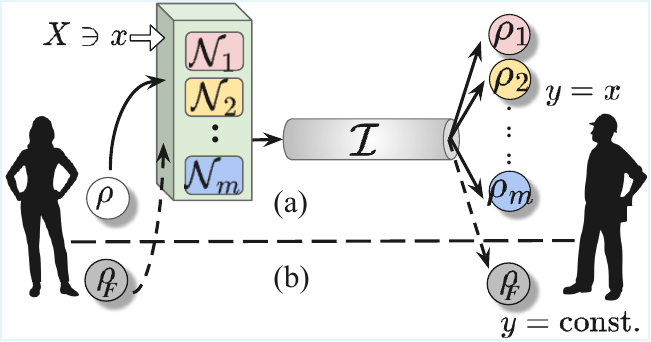}
    \caption{(Color online) \textit{The FEC for pointed QRTs.} (a)  In a pointed QRT \(\mathcal{R}\), Alice can use a resource state \(\rho\) (the clean, white ball) to encode classical information \(x\in X\) by applying different free operations (different coloring), and by decoding it (seeing the color of the ball), Bob can obtain \(y\) (the same color here). (b) On the other hand, the free state \(\rho_F\) (the muddy, gray ball) is unable to communicate any information by remaining invariant (the same gray color) under all possible free operations (different coloring).}
    \label{f0}
\end{figure}

Let us begin with the following lemma, which will be instrumental for the rest of the analysis.
\begin{lemma}\label{l1}
    For any QRT \(\mathcal{R}\), Eq. (\ref{e2}) can be rewritten as
    \begin{align}\label{e3}
    \mathcal{C}_{\mathcal{R}}(\rho)=\max_{p_x,\mathcal{N}_x\in\mathbf{N}_{\mathcal{R}}}\sum_xp_x D(\mathcal{N}_x(\rho)||\overline{\rho}),
    \end{align}
    where, \(\overline{\rho}=\sum_x p_x\mathcal{N}_x(\rho)\) and \(D(\rho||\sigma)\) is the relative entropy between two quantum states \(\rho\) and \(\sigma\).
\end{lemma}
\noindent
An intuition behind Eq. (\ref{e3}) can be found in \cite{vedral2002role}, however, for our proposed settings, we detail the proof in Appendix \ref{pl1}.

We will now examine the salient features of the quantity \(\mathcal{C}_{\mathcal{R}}(\rho)\) to be a valid resource measure.
\begin{proposition}\label{p1}
    (Bounded) For any restricted set of encoding operations \(\mathbf{N}_R\) and for any arbitrary initial quantum state \(\rho\in\mathcal{D}(\mathcal{H}_d)\),
    \[\mathcal{C}_{R}(\rho)\leq\log d.\]
\end{proposition}
\noindent
The proof, although quite intuitive, is given in Appendix \ref{pp1}.
\begin{proposition}\label{p2}
   (Monotonicity) For any QRT \(\mathcal{R}\) and an arbitrary quantum state \(\rho\) 
   \[\mathcal{C}_{\mathcal{R}}(\rho)\geq\mathcal{C}_{\mathcal{R}}(\mathcal{N}(\rho)),\]
   for every \(\mathcal{N}\in\mathbf{N}_{\mathcal{R}}\).
\end{proposition}
\noindent
The proof is deferred in Appendix \ref{pp2}.

We will now show that the FEC \(\mathcal{C}_{\mathcal{R}}(\rho)\), in fact, is a strongly monotonic resource measure under stochastically free operations for any arbitrary QRT. 

\begin{proposition}\label{p3}
  (Strong Monotonicity under stochastically free operations) For an arbitrary QRT \(\mathcal{R}\) and for an arbitrary quantum state \(\rho\) 
  \[\mathcal{C}_{\mathcal{R}}(\rho)\geq\sum_kp_k\mathcal{C}_{\mathcal{R}}(\sigma_k^{\rho}),\]
  where, \(\forall k,~\sigma_k^{\rho}=\frac{\mathcal{E}_k(\rho)}{\Tr[\mathcal{E}_k(\rho)]}\), \(p_k=\Tr[\mathcal{E}_k(\rho)]\) and \(\mathcal{E}_k(*)=L_k(*)L_k^{\dagger}\) is a CP trace non-increasing free operation, and \(\sum_k\mathcal{E}_k(*)=\sum_kL_k(*)L_k^{\dagger}\) is a free operation (CPTP map).
\end{proposition}
\noindent
The proof is detailed in Appendix \ref{pp3}. 

At this point it is important to mention that while the monotonicity in Proposition \ref{p2} holds for any arbitrary free operations, the strong monotonicity in Proposition \ref{p3} only refers to the free operations for which all the Kraus operators are indeed free, i.e., keeping the set of free states \(\mathcal{F}_{\mathcal{R}}\) invariant. For illustration, consider the Gibbs preserving operations in the context of quantum thermodynamics. Although every Gibbs-preserving operation preserves the thermal state, all of them are not stochastically free operations, i.e., the individual Kraus operators are not Gibbs-preserving at all (see the prototypical example in Ref. \cite{faist2015gibbs}). 

\begin{proposition}\label{p4}
    (Convexity) For any QRT \(\mathcal{R}\) and for any set of quantum states \(\{\rho_k\}_{k=0}^{n-1}\),
    \[\mathcal{C}_{\mathcal{R}}(\sum_kp_k\rho_k)\leq\sum_kp_k\mathcal{C}_{\mathcal{R}}(\rho_k),\]
    where, \(\sum_k p_k=1\) with \(0\leq p_k\leq 1\) for every \(k\).
\end{proposition}
\noindent
We defer the proof in Appendix \ref{pp4}.

The above Propositions (\ref{p1} -- \ref{p4}) do not explicitly involve the structure of the set of free states \(\mathcal{F}_{\mathcal{R}}\) for the QRT \(\mathcal{R}\). Therefore, from these four properties we can conclude the following:
\begin{corollary}\label{c2}
    For any arbitrary QRT \(\mathcal{R}\) and for a given quantum state \(\rho\), the free encoding capacity \(\mathcal{C}_{\mathcal{R}}(\rho)\) is a bounded, strongly monotonic, and convex resource measure. 
\end{corollary}
The boundedness of Proposition \ref{p1} implies a suitable scaling for the quantity FEC depending upon the QRT, such that \(\mathcal{C}_{\mathcal{R}}(\rho)=0\) for \(\rho\in\mathcal{F}_{\mathcal{R}}\). However, it does not imply faithfulness, i.e., \(\mathcal{C}_{\mathcal{R}}(\rho)=0\nRightarrow\rho\in\mathcal{F}\) in general. At this juncture, we will revisit the instance of the pointed QRTs.
\begin{proposition}\label{p5}
    (Faithfulness for pointed QRT) For every pointed QRT, the free encoding capacity is a faithful resource measure.
\end{proposition}
\noindent
While the proof is depicted in Appendix \ref{pp5}, in Fig. \ref{f0} we give a schematic idea behind the proof.

Finally, the Propositions \ref{p1} - \ref{p5} together conclude the proof of Theorem \ref{t1}. 

\textit{An illustrative example.} -- We will now discuss an elegant example of pointed QRT, namely the resource theory of purity. For such a pointed QRT over \(d\)-dimensional quantum system, the free state is the maximally mixed one, i.e., \(\rho_F=\frac{\mathbb{I}}{d}\). Within such a QRT, consider the FEC of an arbitrary quantum state \(\rho\in\mathcal{D}(\mathcal{H}_d)\),
\[\mathcal{C}_{\text{purity}}(\rho)=\max_{p_x,\mathcal{N}_x}[S(\sum_xp_x\mathcal{N}_x(\rho))-\sum_xp_xS(\mathcal{N}_x(\rho))],\]
where, every \(\mathcal{N}_x\in\mathbf{N}_{\text{purity}}\) -- the set of unital channels. Now, observe that being a \(d\)-level quantum system
\begin{align}\label{e5}
S(\overline{\rho})=S(\sum_xp_x\mathcal{N}_x(\rho))\leq\log d.
\end{align}
On the other hand, since under every unital channel the entropy of an arbitrary quantum state is always non-decreasing, we have
\begin{align}\label{e6}
\sum_xp_xS(\mathcal{N}_x(\rho))\geq S(\rho).
\end{align}
Interestingly, the unitary operations are indeed the specific choice of free encoding for an arbitrary quantum state \(\rho\), which saturate both the Eqs. (\ref{e5}) and (\ref{e6}) simultaneously. In particular,
for every \(d\)-level quantum state \(\rho\), there exist \(d\) unitary operations \(\{U_k\}_{k=0}^{d-1}\), such that 
\[\frac1d\sum_{k=0}^{d-1}U_k\rho U_k^{\dagger}=\frac{\mathbb{I}}d.\]
If \(\rho=\sum_{i=0}^{d-1}q_i\ketbra{\psi_i}{\psi_i}\), then a canonical choice of \(\{U_k\}_{k=0}^{d-1}\) are the all possible even permutations over the orthonormal basis \(\{\ket{\psi_i}\}_{i=0}^{d-1}\). We can therefore identify the FEC of an arbitrary quantum state \(\rho\in\mathcal{D}(\mathcal{H}_d)\) as
\[\mathcal{C}_{\text{purity}}(\rho)=S(\frac{\mathbb{I}}d)-S(\rho)=\log d-S(\rho).\]

By identifying the unitary operations -- which are the extreme points on the set of unital channels -- as the optimal free encoding, one may be tempted to ask whether the FEC of a quantum state \(\rho\) is always achievable by the extreme free operations of any arbitrary QRT? In the following, we will answer this question affirmatively, and the proof is depicted in Appendix \ref{pt2}.
\begin{theorem}\label{t2}
    For any QRT \(\mathcal{R}\) and for any arbitrary quantum state \(\rho\), the optimal ensemble of free encoding, achieving the FEC, is limited within the extreme points of the free operations of \(\mathcal{R}\).
\end{theorem}

Theorem \ref{t2} can also be interpreted as the encoding analogue of the optimization of the traditional Holevo quantity for an ensemble of quantum state \(\{p_x,\rho_x\}\). Precisely, the accessible information of any quantum channel is the optimized Holevo quantity for all possible input ensemble to the channel. However, it is always sufficient to consider the optimization over the pure quantum states only \cite{holevo2019quantum}.   

\textit{Beyond Pointed QRTs.} -- While the FEC has been established as a potential resource measure for any QRT, we have seen that for the pointed one, it naturally emerges as a faithful one. But what about other resource theories? Is there a suitable rescaling of FEC such that it could be a potentially faithful one? To answer this question, we will first show that faithfulness cannot be achieved for every possible QRT in general.
\begin{proposition}\label{p6}
    For any QRT \(\mathcal{R}\), containing a complete orthonormal basis \(\{\ket{\psi_i}\}_{i=0}^{d-1}\) in the set of free states \(\mathcal{F}_{\mathcal{R}}\), FEC cannot be a faithful resource measure.
\end{proposition}
\begin{proof}
    The proof is trivial. Let us consider an ensemble of free operations \(\{p_x=\frac 1d, \mathcal{N}_x\}\), where \(\mathcal{N}_x(\rho)=\ketbra{\psi_x}{\psi_x},~\forall\rho\). For every possible input quantum state \(\rho\in\mathcal{D}(\mathcal{H}_d)\), using this ensemble of free encoding, Alice can send \(d\)-orthogonal quantum states to Bob and hence
    \[\mathcal{C}_{\mathcal{R}}(\rho)=\log d.\]
    Therefore, any rescaling of FEC for the free states, will also trivialize the measure as a constant for all quantum states.
\end{proof}
While there are several examples of such QRTs, viz., coherence, entanglement, falling under this category, one could have another class of QRTs which are neither pointed, nor containing any full orthonormal basis as the free states. For example, consider the resource theory of thermodynamics in absence of a thermal bath --- namely the activity \cite{swati2023resource}. In Appendix \ref{act}, we have shown that under suitable scaling of FEC, it is in general possible to come up with a valid resource measure for the resource theory of activity.

\textit{Conclusions.} -- We have established that the elegant notion of classical information processing by quantum systems emerges as a universal unit of quantum resources. Moreover, whenever the resource theory admits a unique free state, the resource unit serves as a faithful one. This suggests that any instance of classical information processing via quantum transmission lines costs several quantum resources. In fact, for every practical premise of classical communication via quantum channels one should associate a list of costs for each individual resource-theoretic signatures the quantum system bears. This draws a natural conclusion about the true resource-theoretic cost for communicating an elegant classical information in the real-life scenario: classical communication never comes as free. In other way, it demonstrates a clear separation between a classical channel and a quantum channel, enabling communication of classical information only. This further instigates an interesting question about the resource-cost to stimulate any quantum channel with no private or quantum capacity, however allowing successful classical communication. 

Besides its information-theoretic implications, our results pave the way for a new outlook on characterizing quantum resources universally. One of the most stimulating directions stemming from our work involves the question of the additivity of a quantum resource measure. More precisely, while our analysis mimics the Holevo settings, i.e., with product encoding - joint decoding scenario for the classical information processing, the question remains open to investigate the instances of Hastings-like scenario, under both the joint encoding-decoding settings \cite{hastings2009superadditivity}. Addressing this question, consequently, could answer whether the FEC for a given QRT is super-additive or not. This, in turn, would highlight about the possibilities of decreasing the consumed quantum resources to communicate classical information jointly. On the other hand, while our analysis reveals the potential of FEC to be a perfect resource measure for any pointed resource theory and some of the other QRTs, e.g., activity, it would be interesting to impose physically motivated constraints over the encoding operations further to extend it faithfully over arbitrary quantum resource theories.

\textit{Acknowledgments.}-- TG is supported by the Slovak Research and Development Agency through Grant no. APVV-22-0570, the Scientific Grant Agency of the Ministry of Education, Slovak Republic through Grant no. VEGA 2/0128/24 and the Štefan Schwarz Support Fund 2025/OV1/046 by the Slovak Academy of Sciences.
\bibliography{Bibliography}
\section*{Appendix}
\subsection{Proof of Corollary \ref{c1}}\label{pc1}
The proof is  a simple implication of Proposition \ref{p3} and here we will use the same notation as therein. 
Let us consider a CPTP map \(\mathcal{E}\) on \(\rho\in\mathcal{D}(\mathcal{H}_d)\), with free Kraus operators \(\{L_j\}_{j=0}^{N-1}\) and \(\sum_{j=0}^{N-1}L_j^{\dagger}L_j=\mathbb{I}_{d}\), that is, \(\forall j,~L_j\rho_{F} L_j^{\dagger}\propto\rho_{F}\) with \(\rho_{F}\) being the free state. Then, we can write
\begin{align}
\nonumber\mathcal{E}(\rho)&=\sum_{j\neq k}L_j\rho L_j^{\dagger}+L_k\rho L_k^{\dagger}\\\nonumber&=(1-p_k)\tau^{\rho} + p_k \sigma_k^{\rho},
\end{align}
where, \(\tau^{\rho}=\frac{\sum_{j\neq k}L_j\rho L_j^{\dagger}}{\sum_{j\neq k}\Tr(L_j\rho L_j^{\dagger})}\) and \(\sigma^{\rho}_k=\frac{L_k\rho L_k^{\dagger}}{\Tr(L_k\rho L_k^{\dagger})}\) with \(p_k=\Tr(L_k\rho L_k^{\dagger})\).

Proposition \ref{p3} then implies that
\begin{align}
\nonumber&\mathcal{C}_{\mathcal{R}}(\rho)\geq p_k\mathcal{C}_{\mathcal{R}}(\sigma_k^{\rho})\\\nonumber\implies& p_k\leq \frac{\mathcal{C}_{\mathcal{R}}(\rho)}{\mathcal{C}_{\mathcal{R}}(\sigma_k^{\rho})}
\end{align}
Therefore, any quantum state \(\rho\) can be transformed to \(\sigma\), under the stochastic free operations with a probability no more than \(\frac{\mathcal{C}_{\mathcal{R}}(\rho)}{\mathcal{C}_{\mathcal{R}}(\sigma)}\). Using the trivial probability bound, we can then further rewrite,
\[\Pr(\rho\xrightarrow[\mathcal{R}]{}\sigma)\leq\min\Bigl\{ \frac{\mathcal{C}_{\mathcal{R}}(\rho)}{\mathcal{C}_{\mathcal{R}}(\sigma)}, 1\Bigr\}\]
Note that for the class of resource theories \(\mathcal{R}\) where a complete set of orthonormal basis is included in \(\mathcal{F}_{\mathcal{R}}\), the above bound is trivial, since \(\mathcal{C}_{\mathcal{R}}(\rho)\) is constant for every quantum state \(\rho\). 

\subsection{Proof of Lemma \ref{l1}}\label{pl1}
Let us begin with the identity for the relative entropy between two quantum states \(\mathcal{N}_x(\rho)\) and \(\overline{\rho}\). Since \(\overline{\rho}=\sum_xp_x\mathcal{N}_x(\rho)\), we have \(\text{Supp}(\mathcal{N}_x(\rho))\subseteq\text{Supp}(\overline{\rho})\) for every \(x\). This further implies
\begin{align*}
   D(\mathcal{N}_x(\rho)||\overline{\rho})=-S(\mathcal{N}_x(\rho))-\Tr(\mathcal{N}_x(\rho)\log\overline{\rho}).
\end{align*}
Therefore,
\begin{align}\label{e4}
\nonumber\sum_xp_xD(&\mathcal{N}_x(\rho)||\overline{\rho})\\\nonumber&=-\sum_xp_xS(\mathcal{N}_x(\rho))-\sum_x p_x\Tr[\mathcal{N}_x(\rho)\log\overline{\rho}]\\\nonumber&=-\sum_xp_xS(\mathcal{N}_x(\rho))-\Tr[\sum_xp_x\mathcal{N}_x(\rho)\log \overline{\rho}]\\\nonumber&=-\sum_xp_xS(\mathcal{N}_x(\rho))-\Tr(\overline{\rho}\log\overline{\rho})\\&=-\sum_xp_xS(\mathcal{N}_x(\rho)) + S(\overline\rho)
\end{align}
Finally, we complete the proof by replacing Eq. (\ref{e4}) in Eq. (\ref{e2}).

\subsection{Proof of Proposition \ref{p1}}\label{pp1}
    The proof is rather straightforward. Thanks to the seminal result by Holevo \cite{Holevo1973BoundsFT}, a lone \(d\)-level quantum system cannot communicate more than \(\log d\)-bits of classical information. In fact, for every possible input quantum state \(\rho\in\mathcal{D}(\mathcal{H}_d)\) and a perfect \(d\)-level quantum transmission line, consider \(d\) number of encoding operations \(\{\mathcal{N}_j\}\), such that, \(\mathcal{N}_j(\rho)=\ketbra{j}{j},~\forall\rho\). Now, by choosing all these states \(\{\ket{j}\}_{j=0}^{d-1}\) to be orthogonal, Bob can distinguish them perfectly by a suitable choice of the measurement and hence can extract exactly \(\log d\)-bits of information.  Therefore, for any restricted class of encoding operations \(\mathbf{N}_R\), we have
    \[\mathcal{C}_R(\rho)\leq\mathcal{C}(\rho)=\log d,\]
    where, \(\mathcal{C}(\rho)\) denotes the usual classical capacity with unrestricted encoding operations.

\subsection{Proof of Proposition \ref{p2}}\label{pp2}
    Let us consider the ensemble \(\{q_x^*,\mathcal{M}_x^*\}\) of quantum operations, where \(\mathcal{M}_x^*\in\mathbf{N}_{\mathcal{R}},~\forall x\), to be the optimized encoding for the input state \(\mathcal{N}(\rho)\), achieving the FEC \(\mathcal{C}_{\mathcal{R}}(\mathcal{N}(\rho))\). Therefore,
    \[\mathcal{C}_{\mathcal{R}}(\mathcal{N}(\rho))=S(\sum_x q_x^*\mathcal{M}_x^*(\mathcal{N}(\rho)))-\sum_xq_x^*S(\mathcal{M}_x^*(\mathcal{N}(\rho))).\]
    Any pair of free operations of the QRT \(\mathcal{R}\), must be free under concatenation, i.e., \(\mathcal{M}_x\circ\mathcal{N}\in\mathcal{N}_{\mathcal{R}},\) whenever \(\mathcal{M}_x,~\mathcal{N}\in\mathcal{N}_{\mathcal{R}}\). Therefore,
    \begin{align*}
        \mathcal{C}_{\mathcal{R}}(\mathcal{N}(\rho))&= S(\sum_x q_x^*\mathcal{M}_x^*\circ\mathcal{N}(\rho))-\sum_xq_x^*S(\mathcal{M}_x^*\circ\mathcal{N}(\rho))\\&=S(\sum_x q_x^*\Tilde{\mathcal{N}_x}(\rho))-\sum_xq_x^*S(\tilde{\mathcal{N}_x}(\rho))\\&=\mathcal{C}_{\mathcal{R}}(\rho,\{q_x^*,\Tilde{\mathcal{N}_x}\})
        \\&\leq \mathcal{C}_{\mathcal{R}}(\rho)
    \end{align*}
    Here, the operation \(\tilde{\mathcal{N}_x}=\mathcal{M}_x^*\circ\mathcal{N}\) is a free operation for all \(x\). This completes the proof.

\subsection{Proof of Proposition \ref{p3}}\label{pp3}
Let us first consider a classical information processing scenario involving a \(d\)-dimensional quantum channel \(\mathcal{N}_d\), along with a \(d'\)-level perfect classical channel \(I_{d'}\). If the random variables associated with the channels are denoted as \(X\) and \(Y\) respectively, then the joint capacity of the channels will take the form 
\begin{align}\label{e9}
    \nonumber\mathcal{C}(\mathcal{N}_d\otimes I_{d'})&=\max_{\mu_x,\rho_x,q_y}\chi(\sum_{x,y}\mu_xq_y\mathcal{N}_d(\rho_x)\otimes \ketbra{y}{y})\\&=\max_{\mu_x,\rho_x}\chi(\sum_x\mu_x\mathcal{N}_d(\rho_x))+\max_{q_y}H(q_y).
\end{align}
Here, we have used the notation \(\{\ketbra{y}{y}\}_{y=0}^{d'-1}\) to represent the classical random variable \(Y\) in terms of orthogonal quantum states. Also Eq. (\ref{e9}) can be seen as a special case of the additivity of classical capacity for any quantum channel with an entanglement-breaking channel \cite{shor2002additivity}. Interestingly, it also implies that the joint capacity of the channels involves individual maximization of the quantum and the classical channel. 

Therefore, if the classical channel intends to send a random variable following a specific probability distribution \(\{p_k\}\), then the maximal amount of classical information one can jointly communicate in such a scenario is \(\mathcal{C}(\mathcal{N}_d)+H(p_k)\). This further implies that in the context of free encoding for a given input quantum state \(\rho\), along with a specific classical random variable following the probability distribution \(\{p_k\}\), the maximum amount of information that can be communicated is given by
\begin{align}\label{e10}
\mathcal{I}_{\max}(\rho)=\mathcal{C}_{\mathcal{R}}(\rho)+H(p_k),
\end{align}
where, \(\mathcal{C}_{\mathcal{R}}(\rho)\) is the FEC of the quantum state \(\rho\).

Let us now consider an isometry \(\mathcal{V}:\mathcal{H}_d\to\mathcal{H}_d\otimes\mathcal{H}_{d^2}\), such that for every \(\ket{\psi}\in\mathcal{H}_d\),
\[\mathcal{V}(\ket{\psi})\mapsto\sum_{k=0}^{d^2-1}L_k\ket{\psi}_1\otimes\ket{k}_2.\]
If the projective measurement \(\mathcal{M}_k=\{\ketbra{k}{k}\}_{k=0}^{d^2-1}\) is then performed on the second system and the system is replaced with a classical register, then for the input quantum state \(\rho\in\mathcal{D}(\mathcal{H}_d)\), the evolved quantum-classical state becomes 
\[\sigma_{qc}=\sum_{k=0}^{d^2-1} p_k(\sigma_k^{\rho})_1\otimes\ketbra{k}{k}_2,\]
where, \(\sigma_k^{\rho}=\frac 1{p_k}L_k\rho L_k^{\dagger}\) and \(p_k=\Tr(L_k\rho L_k^{\dagger})\). 

Now, if the Kraus operators \(\{L_k\}\) are such that \(L_k\rho_FL_k^{\dagger}\in\mathcal{F}_{\mathcal{R}}\), for every \(\rho_F\in\mathcal{F}_{\mathcal{R}}\), then the sequence of operations \(\Tr_2\circ(\mathbb{I}_1\otimes{\mathcal{M}_k}_2)\circ\mathcal{V}\) can be treated as a valid free operation on the input quantum state \(\rho\).
    \begin{figure}
    \centering
    \includegraphics[width=1.15\linewidth]{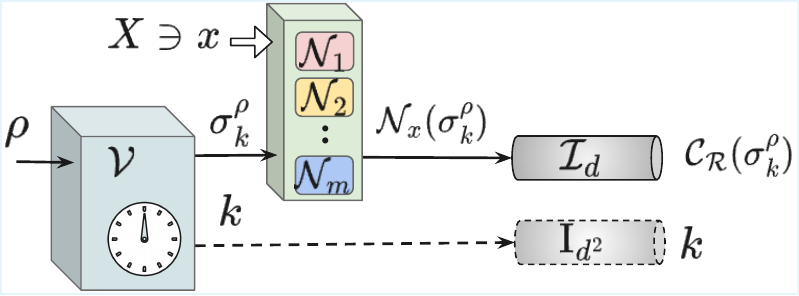}
    \caption{(Color online)\textit{Strong monotonicity of FEC under stochastic free operations.} The input quantum state \(\rho\), after the isometry \(\mathcal{V}\) and the post-selection of the \(k^{\th}\) outcome on the ancilla becomes \(L_k\rho L_k^{\dagger}\propto\sigma_k^{\rho}\). Then the input information \(x\in X\) is encoded on \(\sigma_k^{\rho}\) by a free operation \(\mathcal{N}_x\) and sent through the quantum channel \(\mathcal{I}_d\). As supplementary information, the classical index \(k\) is transmitted through the classical channel \(\text{I}_{d^2}\), shown with the dotted line. The receiver, having the classical index \(k\) in hand, can extract optimally \(\mathcal{C}_{\mathcal{R}}(\sigma_k^{\rho})\) amount of information from the quantum channel.}
    \label{f3}
\end{figure}
Therefore, the sender, Alice, can communicate the symbols \(\{k\}_{k=0}^{d^2-1}\) via a \(d^2\)-dimensional classical channel, and for each of the symbols \(k\in\{0,\cdots,d^2-1\}\), she can utilize the quantum channel to send the state \(\sigma_k^{\rho}\) by applying an ensemble of free operations \(\{q_x^k,\mathcal{N}_x^k\}\) (see Fig. \ref{f3}). Under such a scenario, the effective quantum-classical state Alice can communicate takes the form
\[\tau_{qc}=\sum_{k,x}p_kq_x^k(\mathcal{N}_x^k(\sigma_k^{\rho}))_1\otimes\ketbra{k}{k}_2.\]
Accordingly, the amount of information Bob obtains from the state \(\tau_{qc}\) is given by
\begin{align*}
    \nonumber\chi(\tau_{qc})&=S(\tau_{qc})-\sum_{k,x}p_kq_x^kS(\mathcal{N}_x^k(\sigma_k^{\rho}))\\\nonumber&=H(p_k)+\sum_kp_k S(\sum_xq_x^k\mathcal{N}_x^k(\sigma_k^{\rho}))-\sum_{k,x}p_kq_x^kS(\mathcal{N}_x^k(\sigma_k^{\rho}))\\&=H(p_k)+\sum_kp_k\chi(\sum_xq_x^k\mathcal{N}_x^k(\sigma_k^{\rho})).
\end{align*}
Maximizing over all possible free encoding ensembles \(\{q_x^k,\mathcal{N}_x\}\), we obtain the maximum amount of information that can be conveyed to Bob
\begin{align}\label{e11}
\nonumber\max_{q_x^k,\mathcal{N}_x^k}\chi(\tau_{qc})&=H(p_k)+\sum_kp_k\max_{q_x^k,\mathcal{N}_x^k}\chi(\sum_xq_x^k\mathcal{N}_x^k(\sigma_k^{\rho}))\\&=H(p_k)+\sum_kp_k\mathcal{C}_R(\sigma_k^{\rho}).
\end{align}
Recall that the derivation of Eq. (\ref{e11}), involves only the possible compositions of various free operations on the input quantum state \(\rho\) and a \(d^2\)-dimensional classical channel to communicate a classical random variable with probability distribution \(\{p_k\}\). Therefore, from Eq. (\ref{e10}) we can conclude,
\begin{align*}
    \mathcal{I}_{\max}(\rho)&\geq\max_{q_x^k,\mathcal{N}_x^k}\chi(\tau_{qc})\\\implies\mathcal{C}_{\mathcal{R}}(\rho)&\geq\sum_kp_k\mathcal{C}_R(\sigma_k^{\rho}).
\end{align*}
\subsection{Proof of Proposition \ref{p4}}\label{pp4}
    The proof involves the expression we have derived in Lemma \ref{l1} and the joint convexity of the relative entropy. Precisely, if \(\{q_x^*,\mathcal{N}_x^*\}\) be the encoding ensemble which optimizes the FEC of the state \(\sum_kp_k\rho_k\), we define 
    \begin{align*}
        \overline{\sigma}&:=\sum_{x}q_x^*\mathcal{N}_x^*(\sum_kp_k\rho_k)\\&=\sum_xq_x^*\sum_kp_k\mathcal{N}_x^*(\rho_k)\\&=\sum_kp_k\sum_xq_x^*\mathcal{N}_x^*(\rho_k)\\&=\sum_kp_k\overline{\sigma}_k,
    \end{align*}
    where, \(\overline{\sigma}_k:=\sum_xq_x^*\mathcal{N}_x^*(\rho_k)\).
    Therefore, we can write
    \begin{align*}
        \mathcal{C}_{\mathcal{R}}(\sum_kp_k\rho_k)&=\sum_x q_x^*D(\mathcal{N}_x^*(\sum_kp_k\rho_k)||\overline{\sigma})\\&=\sum_xq_x^*D(\sum_kp_k\mathcal{N}_x^*(\rho_k)||\sum_kp_k\overline{\sigma_k})\\&\leq\sum_xq_x^*\sum_kp_kD(\mathcal{N}_x^*(\rho_k)||\overline{\sigma_k})\\&
        =\sum_kp_k\sum_xq_x^*D(\mathcal{N}_x^*(\rho_k)||\overline{\sigma_k})\\&=\sum_k p_k\mathcal{C}_{\mathcal{R}}(\rho_k,\{q_x^*,\mathcal{N}_x^*\})\\&\leq\sum_kp_k\mathcal{C}_{\mathcal{R}}(\rho_k).
    \end{align*}
    Here, the second equality follows by adopting the linearity of the CPTP maps and the first inequality involves the joint convexity of the quantum relative entropy. The expression \(\mathcal{C}_{\mathcal{R}}(\rho_k,\{q_x^*,\mathcal{N}_x^*\})\), in the fourth equality, denotes the Holevo quantity of the state \(\rho_k\), under the encoding ensemble \(\{q_x^*,\mathcal{N}_x^*\}\) (see Lemma \ref{l1}). Finally, the last inequality follows from the fact that the quantity \(\mathcal{C}_{\mathcal{R}}(\rho_k)\) involves maximization over all possible ensembles of the free operations \(\{p_x,\mathcal{N}_x\in\mathbf{N}_{\mathcal{R}}\}\).

\subsection{Proof of Proposition \ref{p5}}\label{pp5}
    Since every free operation \(\mathcal{N}_x\in\mathbf{N}_f\) of the pointed QRT \(\mathcal{R}\) preserves the unique free state \(\rho_F\), no free encoding on Alice's side can change the given state \(\rho_F\). Hence, no information can be communicated from Alice to Bob, making \(\mathcal{C}_{\mathcal{R}}(\rho_F)=0\).

    On the other hand, note that for any such QRT, there exists a free operation pinning every quantum state to \(\rho_F\), i.e., \(\mathcal{N}_0(\rho)=\rho_F,~\forall\rho\). Therefore, for every resourceful quantum state \(\rho\), there exist at least two distinct free operations, \(\mathcal{I}(\rho)=\rho\) and \(\mathcal{N}_0(\rho)=\rho_F\), which produce two distinct states at Bob's side. So, with at least these two free operations, we can have 
    \[\mathcal{C}_{\mathcal{R}}(\rho)\geq \max_p [p D(\rho||\overline{\rho})+(1-p) D(\rho_F||\overline{\rho})]>0,\]
    where, \(\overline{\rho}=p\rho+(1-p)\rho_F\).

    Therefore, \(\mathcal{C}_{\mathcal{R}}(\rho)=0\Leftrightarrow\rho=\rho_F\) and hence completes the proof.

\subsection{Proof of Theorem \ref{t2}}\label{pt2}
    Let us assume that there exists an optimal ensemble of free encoding operations \(\{p_x,\mathcal{N}_x\in\mathbf{N}_{\mathcal{R}}\}\) to achieve the FEC of the quantum state \(\rho\), such that one of the operations \(\mathcal{N}_z\) is not an extreme free element of the set \(\mathbf{N}_{\mathcal{R}}\), i.e., \(\mathcal{N}_z\notin\text{Ext.}(\mathbf{N}_{\mathcal{R}})\). Then, 
    \[\mathcal{N}_z(\rho)=\sum_kq_z^k\mathcal{N}_z^k(\rho),      \quad \forall\rho,\]
    where \(\forall k,~\mathcal{N}_z^k\in \text{Ext.}(\mathbf{N}_{\mathcal{R}})\).

    Now, observe that
    \begin{align*} D(\mathcal{N}_z(\rho)||\overline{\rho})&=D(\sum_kq_z^k\mathcal{N}_z^k(\rho)||\overline{\rho})\\&\leq \sum_k q_z^kD(\mathcal{N}_z^k(\rho)||\overline{\rho}),
    \end{align*}
    where, \(\overline{\rho}=\sum_xp_x\mathcal{N}_x(\rho)\). Therefore, using Eq. (\ref{e3}) in Lemma \ref{l1}, we can write
    \begin{align*}
    \mathcal{C}_{\mathcal{R}}(\rho)&=\sum_xp_xD(\mathcal{N}_x(\rho)||\overline{\rho})\\&\leq \sum_{x\neq z}p_xD(\mathcal{N}_x(\rho)||\overline{\rho}) + p_z\sum_k q_z^kD(\mathcal{N}_z^k(\rho)||\overline{\rho}),
    \end{align*}
    which leads to a contradiction. Hence the optimal ensemble must belong to the set of extreme points of the free operations.

\subsection{Resource theory of Activity}\label{act}
The resource theory of activity is concerned with the amount of extractable work from an isolated quantum system, i.e., by allowing the unitary actions only. 

Let us consider a \(d\)-level quantum system \(\rho\in\mathcal{D}(\mathcal{H}_d)\), governed by the \(d\)-level system Hamiltonian \(H=\sum_i E_i\ketbra{i}{i}\), where \(E_i\leq E_{i+1}\) and \(E_i\) corresponds to the energy of the energy eigenstate \(\ket{i}\). Moreover, if \(\rho=\sum_k p_k\ketbra{\psi_k}{\psi_k}\) in its spectral representation with \(p_k\geq p_{k+1}\), then it is possible to lower its energy \(E(\rho)=\Tr(\rho H)\) further by applying a suitable unitary \(U=\sum_k \ketbra{k}{\psi_k}\) on it. When this energy is the lowest under any possible unitary evolution, \(U\) then the energy difference 
\[\Tr(\rho H)-\Tr(U\rho U^{\dagger}H)=:W\]
can be quantified as the amount of extractable work, termed as \textit{ergotropy} \cite{allahverdyan2004maximal}. 

Evidently, the energy of the state 
\[U\rho U^{\dagger}=\sum_k p_k \ketbra{k}{k}=:\rho_P\]
cannot be lowered further unitarily and hence it is regarded as a passive state in the literature \cite{lenard1978thermodynamical}.

To put it more formally, any quantum state diagonal in the energy eigenbasis along with an inverse population with the increasing energy levels is a passive state. That is, the state \(\sigma\in\mathcal{D}(\mathcal{H}_d)\), governed by the system Hamiltonian \(H\) (as defined earlier), will be a passive state if and only if it can be written as  
\begin{align}\label{e7}
\sigma=\sum_{k=0}^{d-1}q_k\ketbra{k}{k},\text{ where, } ~q_{k+1}\leq q_k \forall k.
\end{align}

 The set of passive states \(\mathcal{P}_d\) is convex and compact; and any state residing outside the set can be regarded as a resource, namely an \textit{active} state. On the other hand, the free operations for such a resource-theoretic framework can be defined in various ways \cite{singh2021partial, swati2023resource}, however, here we will consider the weakest variant of all such which preserves the set of passive states. Within this set of free operations \(\mathbf{N}_{\text{activity}}\) we will now show that the FEC (with proper scaling) of an arbitrary quantum state can be identified as a faithful resource measure. While our analysis holds for any arbitrary dimension, here we illustrate the qubit case.

We begin with the estimation of the FEC exactly for the free states, i.e., the passive qubits.
\begin{lemma}\label{l2}
    The FEC \(\mathcal{C}_{\text{activity}}(\rho)\simeq0.322\) for all \(\rho\in\mathcal{P}_2\).
\end{lemma}
\begin{proof}
    Using Eq. (\ref{e7}), we can conclude that \(\{\rho_0=\ketbra{0}{0},\rho_1=\frac{\mathbb{I}}2\}\) are the extreme points of the set of passive states \(\mathcal{P}_2\). Therefore, for any \(\rho\in\mathcal{P}_2\) and for any \(\mathcal{N}\in\mathbf{N}_{\text{activity}}\), we have
    \[\mathcal{N}(\rho)=p\rho_0+(1-p)\rho_1,\quad 0\leq p\leq 1.\]

    Thanks to Theorem \ref{t2}, we can now identify the channels \(\mathcal{N}_0(\rho)=\ketbra{0}{0}\) and \(\mathcal{N}_1(\rho)=\frac{\mathbf{I}}2\) --- pinning every quantum states to \(\ket{0}\) and \(\frac{\mathbf{I}}2\) respectively --- as the potential free operations to encode the classical messages \(x\in\{0,1\}\). Therefore,
    \begin{align}
    \nonumber\mathcal{C}_{\text{activity}}(\rho)&=\max_q [S(\overline{\rho_q})-qS(\rho_0)-(1-q) S(\rho_1)]
    \\\nonumber&=\max_q [S(\overline{\rho_q})-(1-q)],
    \end{align}
    where, \(\overline{\rho_q}=q\rho_0+(1-q)\rho_1\), \(\rho_0=\ketbra{0}{0}\) and \(\rho_1=\frac{\mathbb{I}}{2}\).
    A simple numerical exercise then reveals that the optimal value is \(q= 0.6\) and hence, \(\mathcal{C}_{\text{activity}}(\rho)\simeq 0.322\) for every qubit passive state \(\rho\in\mathcal{P}_2\) and denote it as \(\mathcal{C}_{\text{activity}}(\rho_{\text{passive}})\) (see Fig \ref{f2}(a)).
\end{proof}
We will now show that for every active qubit the FEC is strictly greater than the FEC of the passive states.
\begin{proposition}\label{p7}
    For every active state \(\rho_a\notin\mathcal{P}_2\), 
    \[\mathcal{C}_{\text{activity}}(\rho_a)>\mathcal{C}_{\text{activity}}(\rho_p),\]
    where, \(\rho_p\in\mathcal{P}_2\).
\end{proposition}
\begin{proof}
    Let us begin with considering any two-level quantum system diagonal in the \(\{\ket{\pm}\}\) basis --- the eigenvectors of the Pauli operator \(\sigma_X\),
    \begin{align*}
    \rho_{\mu}=\mu\ketbra{+}{+} + (1-\mu) \ketbra{-}{-},\end{align*}
    where, \(\ket{\pm}=\frac{\ket{0}\pm\ket{1}}{\sqrt{2}}\) and \(0\leq\mu\leq 1\).
    For every such, \(\rho_{\mu}\) consider an ensemble of free operations \(\{\mathcal{M}_x\}\equiv\{\mathcal{I},\mathcal{Z},\mathcal{N}_0\}\) with probabilities \(\{p_x\}=\{0.2,0.2,0.6\}\). Here \(\mathcal{I}\) and \(\mathcal{Z}\) are the channel representations of the unitary operators \(\mathbb{I}\) and \(\sigma_Z\) respectively and \(\mathcal{N}_0(\rho)=\ketbra{0}{0},~\forall \rho\in\mathcal{D}(\mathcal{H}_2)\). 
    
    Now, with the specific free encoding ensemble \(\{p_x,\mathcal{M}_x\}\) the FEC of any \(\rho_{\mu}\) becomes
    \begin{align}\label{e8}
        \nonumber\mathcal{C}_{\text{activity}}^*(\rho_{\mu},\{p_x,\mathcal{M}_x\})&=S(\overline{\rho_{\mu}})-\sum_xp_xS(\mathcal{M}_x(\rho_{\mu}))
        \\\nonumber&=S(\overline{\rho_{\mu}})-2\times 0.2\times S(\rho_{\mu})
        \\\nonumber&\geq S(\overline{\rho_{\mu}})-0.4\\&\simeq 0.322,
    \end{align}
    where, the second equality follows from the fact that \(S(\rho)=S(\sigma_Z\rho\sigma_Z)\) and \(S(\ketbra{0}{0})=0\) and then the inequality uses \(S(\rho)\leq 1\) for every qubit state. Finally, we obtain the third equality by identifying
    \begin{align*}
        \overline{\rho_{\mu}}&=2\times0.2(\frac12\rho_{\mu}+\frac12 \sigma_Z\rho_{\mu}\sigma_Z)+0.6\ketbra{0}{0}\\&=0.4\frac{\mathbb{I}}2+0.6\ketbra{0}{0},
    \end{align*}
     that is, the optimal ensemble for all passive states (see Lemma \ref{l2}).
        \begin{figure}
    \centering
    \includegraphics[width=1.0\linewidth]{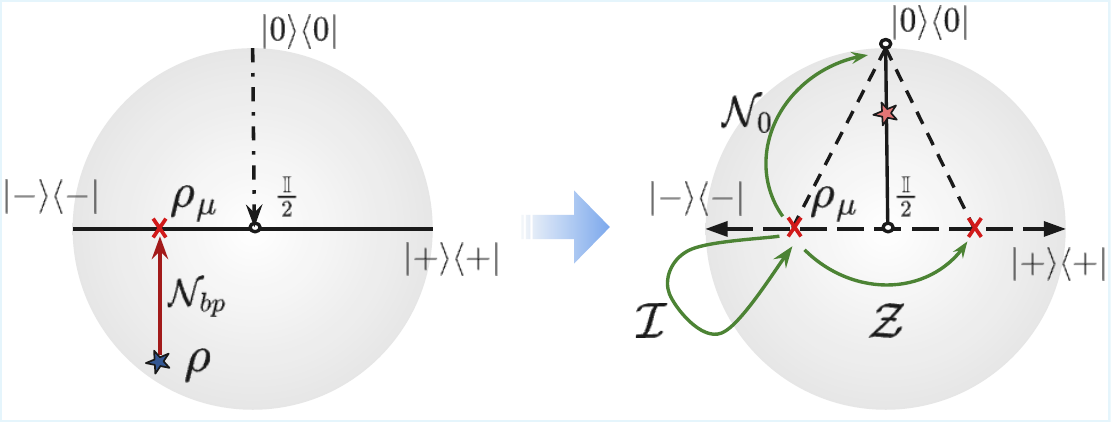}
    \caption{(Color online) \(\mathcal{C}^*_{\text{activity}}\)\textit{ for qubit states coherent in energy eigen basis.} The leftmost image shows that the bit-flip operations \(\mathcal{N}_{bp}\) dephase any quantum state \(\rho\) on the respective version \(\rho_{\mu}\) state on the X-axis, and specifically every passive state to the maximally mixed \(\frac{\mathbb{I}}2\). Then by applying three free operations \(\{\mathcal{I},\mathcal{Z},\mathcal{N}_0\}\), with probabilities \(\{0.2,0.2,0.6\}\) respectively, Alice effectively communicates the state \(\overline{\rho_{\mu}}\), denoted as a \textit{star} symbol. Note that the star is in the same position as in Fig \ref{f2} (a).}
    \label{f1}
\end{figure}

     Note that the inequality in Eq. (\ref{e8}) saturates only for \(\mu=0.5\), i.e., \(\rho_{0.5}=\frac{\mathbb{I}}2\), which is the only passive state of the form \(\rho_{\mu}\). We can therefore conclude 
     \[\mathcal{C}_{\text{activity}}(\rho_{\mu})\geq\mathcal{C}_{\text{activity}}^*(\rho_{\mu},\{p_x,\mathcal{M}_x\})>\mathcal{C}_{\text{activity}}(\rho_{\text{passive}}).\]

    Let us now consider the qubit bit flip channel, 
    \[\mathcal{N}_{bp}(\rho)=\frac 12(\rho + \sigma_X\rho \sigma_X).\] 

    Although the operation \(\mathcal{N}_{bp}\) may increase the amount of ergotropy, it is indeed a passivity-preserving operation. In fact, for any quantum state \(\rho_{\text{diag.}}\) diagonal in the energy eigenbasis (which trivially includes \(\mathcal{P}_2\)), 
    \[\mathcal{N}_{bp}(\rho_{\text{diag.}})=\frac{\mathbb{I}}2\in\mathcal{P}_2.\]
    Moreover, under this channel action every qubit state gets mapped on the X-axis of the Bloch sphere, that is, decohered in the \(\sigma_X\) basis. 
 
    Now, for any arbitrary qubit state \(\rho_{\text{coh}}\) possessing coherence in the \(\sigma_Z\)-basis, one can apply the passivity-preserving operation \(\mathcal{N}_{bp}\) first and then follow the encoding ensemble \(\{p_x,\mathcal{M}_x\}\) to communicate classical information. In other words,
    the effective free encoding ensemble will be of the form \(\{p_x,\mathcal{M}_x\circ \mathcal{N}_{bp}\}\) (see Fig \ref{f1}). Since the composition of any set of free operations is also free, the composed operations \(\mathcal{M}_x\circ \mathcal{N}_{bp}\) are also passivity-preserving in nature. Therefore, 
    \begin{align*}\mathcal{C}_{\text{activity}}(\rho_{\text{coh}})&\geq \mathcal{C}^*_{\text{activity}}(\rho_{\text{coh}},\{p_x,\mathcal{M}_x\circ \mathcal{N}_{bp}\})\\&>\mathcal{C}_{\text{activity}}(\rho_{\text{passive}}).\end{align*}

\begin{figure}[t]
    \centering
    \includegraphics[width=0.75\linewidth]{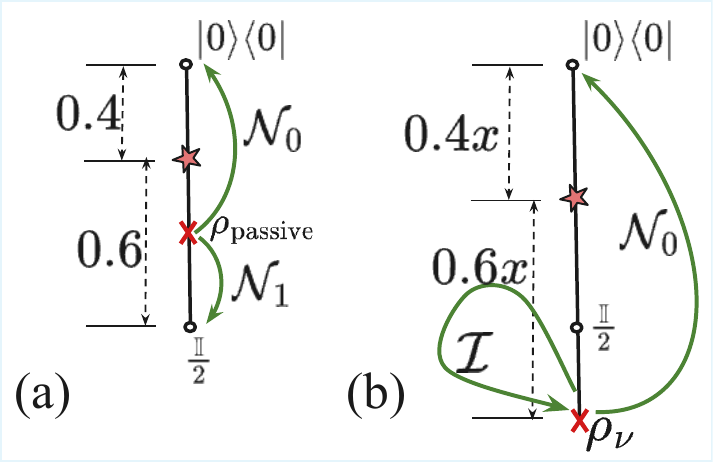}
    \caption{(Color online) \(\mathcal{C}^*_{\text{activity}}\)\textit{ for qubit states incoherent in energy eigen basis.} (a) Any passive state \(\rho_{\text{passive}}\) achieves its FEC with two free operations \(\{\mathcal{N}_0,\mathcal{N}_1\}\), respectively with probabilities \(\{0.6,0.4\}\). This produces an effective state denoted as a \textit{star} symbol on the \(Z\)-axis. (b) For any active state incoherent in the energy eigenbasis, Alice adopts two free operations \(\{\mathcal{N}_0,\mathcal{I}\}\), again with probabilities \(\{0.6, 0.4\}\). This effectively produces the state, denoted by a \textit{star} symbol. Since the distance \(x\) between the states \(\rho_{\nu}\) and \(\ketbra{0}{0}\) is greater than \(1\), it is trivial that the star state in figure (b) possesses higher entropy than that of figure (a).}
    \label{f2}
\end{figure}
    Finally, we will consider the active states, which are incoherent in nature, that is,
    \[\rho_{\nu}=\nu \ketbra{0}{0}+(1-\nu)\ketbra{1}{1},\text{with } 0\leq\nu< \frac 12,\]
    along with a free encoding ensemble \(\{\mathcal{N}_0,\mathcal{I}\}\), with a probability \(\{0.6,0.4\}\) (see Fig \ref{f2}(b)). Accordingly, the FEC of the state \(\rho_{\nu}\) with such an ensemble of free encoding will become
    \[\mathcal{C}_{\text{activity}}^*(\rho_{\nu},\{\{0.6,0.4\},\{\mathcal{N}_0,\mathcal{I}\}\})=S(\overline{\rho_{\nu}})-0.4\times S(\rho_{\nu}),\]
    where, \[\overline{\rho_{\nu}}=0.6\ketbra{0}{0}+0.4\rho_{\nu}=(0.6+0.4\nu)\ketbra{0}{0}+0.4(1-\nu)\ketbra{1}{1}.\]
    Note that, \(S(\overline{\rho_{\nu}})>S(\overline{\rho_{\frac12}})\) and \(S(\rho_{\nu})< S(\rho_{\frac 12})\) for every values of \(\nu\in[0,\frac 12)\). Therefore, 
    \begin{align*}
    \mathcal{C}_{\text{activity}}(\rho_{\nu})&\geq\mathcal{C}_{\text{activity}}^*(\rho_{\nu},\{\{0.6,0.4\},\{\mathcal{N}_0,\mathcal{I}\}\})\\&>S(0.8\ketbra{0}{0}+0.2\ketbra{1}{1})-0.4\\&=\mathcal{C}_{\text{activity}}(\rho_{\text{passive}}).
    \end{align*}
    This completes the proof.
\end{proof}
With the help of Lemma \ref{l2} and Proposition \ref{p7}, we can then conclude that
the quantity \[\tilde{\mathcal{C}}_{\text{activity}}(\rho)=\mathcal{C}_{\text{activity}}(\rho)-0.322\] can be used as a faithful quantifier for all active qubit states.
\end{document}